\documentclass[journal,twoside,web]{ieeecolor}
\usepackage[utf8]{inputenc}
\providecommand{\refname}{References}
\usepackage{generic}
\usepackage{cite}
\usepackage{amsmath,amssymb,amsfonts}
\usepackage{graphicx}
\usepackage[ruled,vlined]{algorithm2e}
\usepackage{hyperref}
\hypersetup{hidelinks=true}
\usepackage{textcomp}
\usepackage{booktabs}
\usepackage{multirow}
\usepackage{placeins}

\newtheorem{proposition}{Proposition}

\def\BibTeX{{\rm B\kern-.05em{\sc i\kern-.025em b}\kern-.08em
    T\kern-.1667em\lower.7ex\hbox{E}\kern-.125emX}}
\providecommand{\refname}{References}
\begin{document}

\title{ Feynman-Kac Reweighted Schrödinger Bridge Matching for Surface-based Tau PET Harmonization}

\author{Jianwei Zhang, Xinyu Nie, Jiaxin Yue, and Yonggang Shi
\thanks{This work was supported by the National Institutes of Health (NIH) under grants R01EB022744, RF1AG077578,  RF1AG064584,  U19AG078109, P30AG066530 and S10OD032285.}
\thanks{Jianwei Zhang, Xinyu Nie, Jiaxin Yue, and Yonggang Shi are with the Stevens Neuroimaging and Informatics Institute, University of Southern California, Los Angeles, CA 90089 USA. Jianwei Zhang, Jiaxin Yue, and Yonggang Shi are with the Ming Hsieh Department of Electrical and Computer Engineering of Viterbi School of Engineering, University of Southern California, Los Angeles, CA 90089 USA. Yonggang Shi is also with the Alfred E. Mann Department of Biomedical Engineering of Viterbi School of Engineering, University of Southern California, Los Angeles, CA, USA.}
\thanks{Corresponding author: Yonggang Shi. Email: yshi@loni.usc.edu.}
}

\maketitle

\begin{abstract}
Tau positron emission tomography (PET)  is widely used for the  \emph{in vivo} characterization of disease stage and progression in Alzheimer's disease (AD). With the adoption of multiple tau PET tracers including AV-1451, PI-2620, MK-6240 with different binding behaviors in various large-scale studies, there is a great need of effective harmonization methods to enable the cross-tracer integration of tau PET datasets. While previous methods such as CenTauR were proposed to standardize scalar tau PET measures, they are limited in accounting for the heterogeneity of tau pathology. In this work, we propose Feynman-Kac Reweighted Schr\"odinger Bridge Matching (FKRSBM), a surface-based framework for cross-tracer tau PET harmonization. FKRSBM learns a direct stochastic transport between tracer domains using Schr\"odinger Bridge matching, avoiding the Gaussian-prior routing used in diffusion-based translation. To promote biologically consistent transport, FKRSBM introduces an endpoint penalty favoring bridge pairings with matched tau-pathology status and implements it through a Feynman-Kac reweighted endpoint proposal. To preserve cortical organization, FKRSBM uses a spherical convolutional network for vertex-level harmonization on cortical surface meshes. In our experiments, we demonstrate our method by harmonizing Tau PET images  acquired with the AV-1451 (n=1480) and PI-2620 (n=2458) tracers from two large-scale datasets. Compared to previous methods including ComBat, CycleGAN, Diffusion Model(DF), and unregularized Schr\"odinger Bridge Model(DSBM), the proposed FKRSBM method outperforms these baselines in subgroup-level alignment, tau-positivity consistency, and  diagnostic classification while preserving subject-specific cortical topography of tau pathology. The code is available at: https://github.com/jianweizhang17/FKRSBM.
\end{abstract}

\begin{IEEEkeywords}
Tau PET, harmonization, Schr\"odinger bridge, Feynman-Kac reweighting, cortical surface.
\end{IEEEkeywords}

\section{Introduction}

\IEEEPARstart{T}{au} positron emission tomography (PET) is an important tool for staging Alzheimer's disease (AD), as tau pathology tracks dementia-related cognitive decline~\cite{braak1991neuropathological, scholl2016pet, ossenkoppele2016tau, hanseeuw2019association}. With the growth of large-scale tau PET datasets, a key challenge is integrating cohorts acquired with different tau tracers, such as AV-1451~\cite{chien2013early}, PI-2620~\cite{kroth2019discovery}, and MK-6240~\cite{hostetler2016preclinical}, which differ in binding affinity, off-target binding profiles, and quantitative dynamic range~\cite{aguero2024head,gogola2022direct}. These tracer-dependent differences introduce systematic non-biological shifts in standardized uptake value ratio (SUVR), inflating biomarker variance, reducing sensitivity to disease effects, and biasing downstream analyses~\cite{varoquaux2022machine}. Although scalar standardization methods such as CenTauR~\cite{ville2023centaur} have been proposed to provide valuable cross-tracer summary measures, they are limited in accounting for the heterogeneity of tau pathology \cite{ville2023centaur,vogel2021four,young2022divergent}. Thus, it remains a challenge to have effective and image-based harmonization of tau PET data across different tracers. 

\begin{figure*}[t]
    \centering
    \includegraphics[width=0.8\linewidth]{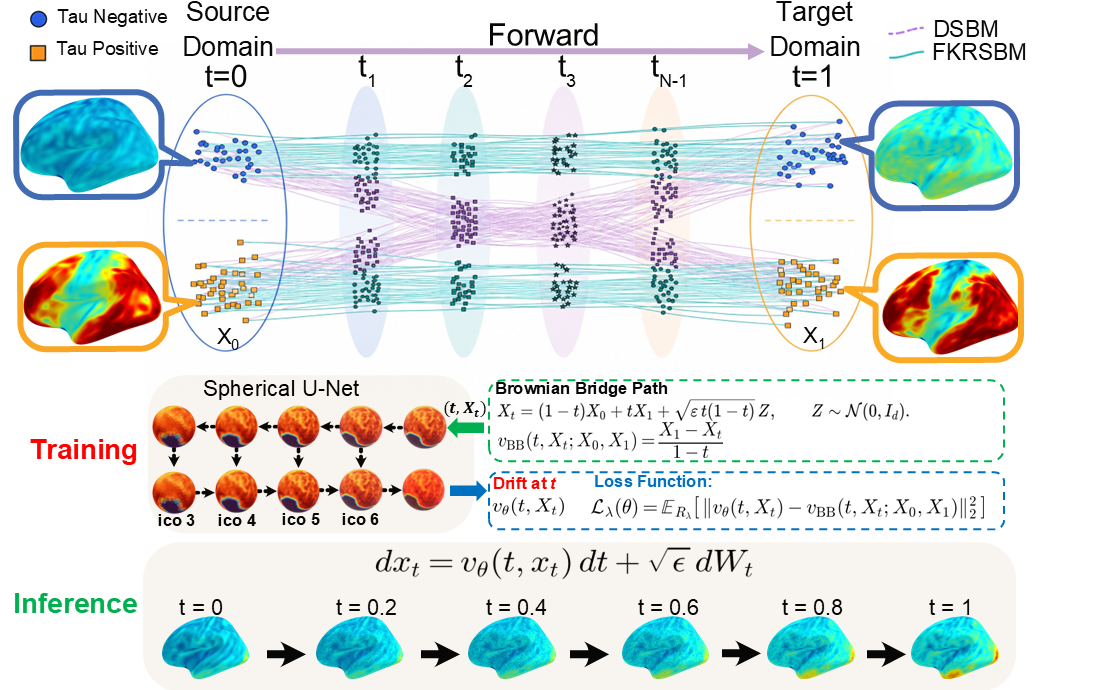}
    \caption{\textbf{Overview of the proposed Feynman-Kac Reweighted Schr\"odinger Bridge (FKRSBM) for cross-tracer tau PET harmonization.}
    \textit{Top:} Schematic of the bridge between source ($X_0$, PI-2620) and target ($X_1$, AV-1451) tau distributions. The DSBM (dashed purple) transports mass between marginals without regard to biological subgroups, producing crossing trajectories that mix subjects across subgroups. FKRSBM (solid teal) reweights the endpoint coupling so that trajectories preferentially connect subjects sharing the same biological subgroup (e.g., tau-positive vs. tau-negative), yielding  subgroup-consistent paths while preserving the path-space dynamics.
    \textit{Middle (Training stage):} At each iteration, a time $t$ and a Brownian-bridge sample $X_t$ are drawn from a Feynman-Kac reweighted endpoint coupling, and a hierarchical surface network operating across icosahedral resolutions (ico~3--6) predicts the drift $v_\theta(t, X_t)$, regressed against the analytic Brownian bridge drift $v_{BB}$.
    \textit{Bottom (Inference stage):} A source SUVR map is harmonized by integrating the learned SDE from $t=0$ to $t=1$, producing a trajectory of intermediate cortical maps that terminates in the target tracer space.}
    \label{fig:model_overview}
\end{figure*}

While paired or traveling-subject data would provide the most direct basis for estimating tracer-specific differences, such data are rarely available in large-scale cohorts~\cite{liu2024learning}. Practical harmonization must therefore operate in an unpaired setting, estimating tracer effects from cohort-level distributions rather than matched scans. 
A range of methods have been proposed for this type of harmonization. \textit{Statistical methods} such as ComBat~\cite{johnson2007combat, FORTIN2018104} model cohort effects as feature-wise additive and multiplicative adjustments estimated by empirical Bayes. Although robust for regional summaries, they are limited in modeling spatially structured, nonlinear, or region-dependent tracer discrepancies~\cite{marquand2016normative, 10.1016/j.neuroimage.2024.120737}. \textit{Generative methods} relax these assumptions but introduce new risks. CycleGAN-based harmonization~\cite{zhu2017cyclegan, dewey2019deepharmony, zhao2019harmonization} enables unpaired translation through cycle consistency, but can be unstable and prone to mode collapse, potentially suppressing biological variability. Disentanglement-based methods~\cite{jha2018cyclevae, moyer2020scanner, zuo2021information, dinsdale2021deep} seek to separate biological content from acquisition or tracer-related variation, but such factorization is generally non-identifiable without strong inductive biases~\cite{locatello2019challenging, burgess2018understanding}. Diffusion-based translation, including surface-based tau harmonization~\cite{yue2025tau} and SDEdit-style noise injection~\cite{meng2021sdedit}, provides flexible image-to-image mapping, but Gaussian noising creates a tradeoff between domain alignment and preservation of fine details. 

The Schr\"odinger Bridge (SB) offers a principled alternative to the previous methods by formulating unpaired translation as entropy-regularized optimal transport between source and target distributions~\cite{debortoli2021dsb, liu2023i2sb, kim2024unpaired}. Unlike diffusion methods that route data through a Gaussian prior, SB learns a \emph{direct} stochastic transport between empirical distributions, making it well suited for harmonization tasks where subject-level structure should be preserved~\cite{shi2023diffusion,de2024schrodinger}. Yet, despite these advantages, standard SB matching can mix biologically distinct subgroups in unpaired tau PET harmonization. Typical Schr\"odinger Bridge matching~\cite{shi2023diffusion, de2024schrodinger} commonly draws endpoints independently from the product coupling $\pi_0\otimes\pi_1$, treating all cross-tracer pairings as equally plausible. Independent endpoint sampling can produce trajectories that mix clinically distinct tau-pathology states, conflating tracer-dependent measurement shifts with the biology that harmonization should preserve. On the other hand, soft-regularized Schr\"odinger bridges~\cite{ma2025softconstrained} penalize mismatched endpoint pairings. Building on this idea, we propose \textbf{Feynman-Kac Reweighted Schr\"odinger Bridge Matching (FKRSBM)} for unpaired cross-tracer tau PET harmonization. Beyond endpoint consistency, tau PET harmonization also requires a representation that preserves cortical topography. Surface-based PET analysis further respects cortical geometry and helps limit contamination from tracer-specific non-cortical off-target signal~\cite{greve2014cortical,aguero2024head}. We therefore design FKRSBM  for surface-based tau PET data, enabling vertex-wise harmonization on cortical meshes. An overview of the proposed FKRSBM framework is shown in Figure~\ref{fig:model_overview}. Our contributions are twofold:
\begin{enumerate}
\item \textbf{Surface-based Schr\"odinger Bridge harmonization.}
We introduce an SB-based harmonization framework that operates directly on cortical surface meshes. A spherical convolutional backbone~\cite{Spherical_unet} respects cortical geometry and performs vertex-level harmonization without reducing the data to region-level summaries.

\item \textbf{Endpoint-regularized transport.}
We formulate an endpoint-regularized SB objective that augments the standard KL divergence minimization with a covariate penalty $\Phi$, reducing the conflation of tracer effects with biological variation. We show that the resulting source-anchored reweighting can be implemented as standard bridge matching under a Feynman-Kac reweighted endpoint proposal obtained by Radon--Nikod\'{y}m reweighting, with $\lambda$ controlling constraint strength.
\end{enumerate}
In our experiments, we evaluate FKRSBM on unpaired tau PET SUVR maps from two large-scale studies: tau PET acquired with PI-2620 tracer from the Health and Aging Brain Study-Health Disparities (HABS-HD)~\cite{petersen2025health} and tau PET acquired with the AV-1451 tracer from  the Alzheimer's Disease Neuroimaging Initiative (ADNI)~\cite{Mueller05}. Compared with previous methods including  ComBat~\cite{FORTIN2018104}, CycleGAN~\cite{zhao2019harmonization}, a surface diffusion model (DF)~\cite{yue2025tau}, and unregularized DSBM~\cite{de2024schrodinger}, FKRSBM achieves stronger subgroup-level distributional alignment, the lowest residual domain separability, and the best tau-subgroup consistency while preserving subject-level cortical patterns of tau pathology. These results demonstrate that Feynman-Kac reweighting improves SB harmonization in practical cross-tracer tau PET experiments.

\section{Method}
\label{sec:method}

\subsection{Problem Formulation}
\label{subsec:problem}

Let $\pi_0$ and $\pi_1$ denote the source and target data
distributions, respectively, with samples represented as
$x \in \mathbb{R}^d$. In the unpaired harmonization setting, no subject-level correspondence is available across domains. The goal is to map samples from $\pi_0$ toward the target distribution $\pi_1$ while reducing acquisition-related variability and preserving
biologically meaningful variation. We formulate this task within the Schr\"odinger Bridge framework. This section first reviews the
standard Schr\"odinger Bridge problem and bridge matching, then introduces a covariate-informed endpoint regularization, establishes
its equivalence to bridge matching under a Feynman-Kac reweighted reference bridge measure, and derives the resulting training and
inference procedures.

\subsection{Background: Schr\"odinger Bridge and Bridge Matching}
\label{subsec:sb_background}

Let $\Omega$ denote the space of continuous trajectories $\{x_t\}_{t\in[0,1]}$ with $x_t\in\mathbb{R}^d$, and let $\mathbb{W}_\epsilon$ be a reference Wiener measure (i.e. probability distribution) with diffusion scale~$\epsilon$. The Schr\"odinger Bridge (SB) problem seeks a path measure $Q^\ast$ that interpolates between $\pi_0$ and $\pi_1$
while remaining as close as possible to $\mathbb{W}_\epsilon$ in KL divergence:
\begin{equation}
\label{eq:sb-original}
Q^\ast = \operatorname*{argmin}_{\substack{Q \in \mathbb{P}(\Omega) \\
Q_0 = \pi_0,\; Q_1 = \pi_1}}
\mathrm{KL}(Q \,\|\, \mathbb{W}_\epsilon).
\end{equation}
The solution $Q^\ast$ can be represented as an It\^{o} diffusion sharing the same diffusion coefficient as the reference process:
\begin{equation}
\label{eq:ito-sde}
dx_t = b^\ast(t, x_t)\,dt + \sqrt{\epsilon}\,dW_t,
\end{equation}
where $W_t$ is standard Brownian motion and $b^\ast(t,x)$ is the optimal drift. Eq.~\eqref{eq:sb-original} selects, among all stochastic processes connecting $\pi_0$ to $\pi_1$, the one that deviates minimally from pure Brownian motion. For harmonization, this minimal-deviation property discourages the transport from inventing cortical structure absent in the input.

The SB solution admits a decomposition that exposes its two constituent parts~\cite{debortoli2021dsb,shi2023diffusion}. When the marginals are point masses $\pi_0=\delta_{x_0}$ and $\pi_1=\delta_{x_1}$, the solution $Q^\ast$ reduces to the \emph{Brownian bridge} connecting $x_0$ to $x_1$: the law of the reference diffusion $\sqrt{\epsilon}\,dW_t$ pinned to start at $x_0$ at $t=0$ and to end at $x_1$ at $t=1$~\cite{oksendal2003stochastic}. Conditioned on its endpoints, this bridge admits the closed-form interpolation
\begin{equation}
\label{eq:brownian-bridge-interp}
x_t = (1-t)\,x_0 + t\,x_1 + \sqrt{\epsilon\,t(1-t)}\;Z, \qquad Z\sim\mathcal{N}(0,I),
\end{equation}

For general non-degenerate marginals, $Q^\ast$ is a \emph{mixture of Brownian bridges}~\cite{shi2023diffusion,de2024schrodinger}. The joint endpoint distribution $(x_0,x_1)$ follows an optimal entropic coupling $\gamma^\ast\in\Pi(\pi_0,\pi_1)$, and, conditional on $(x_0,x_1)\sim\gamma^\ast$, the intermediate trajectory $\{x_t\}_{t\in(0,1)}$ is the Brownian bridge \eqref{eq:brownian-bridge-interp} between those endpoints. This separates the problem into two parts: (i)~learning the endpoint coupling $\gamma^\ast$ between the marginals, and (ii)~the conditional path dynamics, which are analytically tractable Brownian bridges. The only intractable component is the coupling $\gamma^\ast$, since the bridges connecting any fixed endpoint pair are known in closed form. This structure motivates bridge matching~\cite{shi2023diffusion,de2024schrodinger}, which sidesteps direct computation of $\gamma^\ast$ by regressing the transport drift onto the analytic Brownian-bridge dynamics.

Bridge matching~\cite{shi2023diffusion} learns the transport dynamics by regressing a parameterized drift $v_\theta(t,x)$ onto the conditional drift of Brownian bridges connecting samples from the two endpoint marginals. Specifically, for endpoint pairs $(X_0,X_1)$ and an intermediate bridge sample $X_t$, the model is trained at randomly sampled times
$t\in(0,1)$ by minimizing
\begin{equation}
\label{eq:bridge-matching-loss}
\mathcal{L}(\theta)
=
\mathbb{E}\left[
\left\|
v_\theta(t,X_t)
-
v_{\mathrm{BB}}(t,X_t;X_0,X_1)
\right\|_2^2
\right],
\end{equation}
where the Brownian bridge drift conditioned on endpoints $(X_0,X_1)$ is
given by
\begin{equation}
v_{\mathrm{BB}}(t,x;X_0,X_1)
=
\frac{X_1-x}{1-t}.
\end{equation}
The intermediate state $X_t=(1-t)X_0+tX_1+\sqrt{\epsilon t(1-t)}Z$ with $Z\sim\mathcal{N}(0,I)$ is sampled from the corresponding conditional Brownian bridge. The Diffusion Schr\"odinger Bridge Matching (DSBM)
framework~\cite{de2024schrodinger,shi2023diffusion} operationalizes this idea by alternating between forward and backward bridge constructions and fitting the drift via regression. A single regression yields the \emph{Markovian projection} of the bridge mixture, the Markov diffusion matching its time-marginals rather than the Schr\"odinger bridge directly~\cite{shi2023diffusion}. DSBM~\cite{de2024schrodinger,shi2023diffusion} therefore alternates this projection with a reciprocal projection that rebuilds the mixture from the current coupling; its fixed point, being both Markov and reciprocal, is $Q^\ast$, the neural analogue of the IPF/Sinkhorn iterations in the original DSB~\cite{debortoli2021dsb}.

\subsection{Endpoint-Regularized Schr\"odinger Bridge}
\label{subsec:endpoint_regularized}

The DSBM-style bridge matching loss in Eq.~\eqref{eq:bridge-matching-loss}
is commonly estimated by sampling endpoint pairs independently from the product measure
$\pi_0 \otimes \pi_1$. In the context of cross-tracer harmonization, the model treats all cross-cohort pairings as equally plausible. For example, a tau-negative subject from the source cohort is just as likely to be paired with a tau-positive and a tau-negative subject from the target cohort. When cohorts differ in the
composition of tau-positive and tau-negative subjects, such agnostic pairing can bias the learned transport, conflating tracer effects with biological variation.

To address this, we introduce an endpoint penalty $\Phi: \mathbb{R}^d \times \mathbb{R}^d \to \mathbb{R}$ that encodes prior knowledge about which pairings are biologically meaningful. We then define an endpoint-regularized SB objective that augments the standard KL minimization with a penalty on the expected value of $\Phi$ under the learned path measure. We first define a reference path measure $R$ as an independent mixture of Brownian bridges. Let $\mathcal{W}^{x_0 \to x_1}$ denote the Brownian bridge path measure connecting $x_0$ to
$x_1$ with diffusion scale inherited from $\mathbb{W}_\epsilon$.
The reference measure is:
\begin{equation}
\label{eq:indep-bridge-mixture}
R(\cdot) = \!\int_{\mathbb{R}^d \times \mathbb{R}^d}
\!\!\mathcal{W}^{x_0 \to x_1}(\cdot)\,
\pi_0(dx_0)\,\pi_1(dx_1).
\end{equation}
Under $R$, the endpoints are drawn independently:
$(X_0, X_1) \sim \pi_0 \otimes \pi_1$. This is the implicit reference used by product-pairing DSBM~\cite{shi2023diffusion} training. For a penalty strength $\lambda \ge 0$, we seek a path measure that stays close to the
reference $R$ while penalizing biologically implausible endpoint pairings:
\begin{equation}
\label{eq:endpoint-regularized-objective}
P^\star_\lambda \in \operatorname*{argmin}_{\substack{P \ll R \\ P_0 = \pi_0 }}
\Big\{
  \mathrm{KL}(P \,\|\, R)
  + \lambda\,\mathbb{E}_{P}[\Phi(X_0, X_1)]
\Big\},
\end{equation}
where $P \ll R$ denotes that $P$ is absolutely continuous with respect to $R$, and the constraint $P_0 = \pi_0$ fixes the source marginal. Unlike the two-marginal Schr\"odinger bridge of Section~\ref{subsec:sb_background}, we deliberately leave the target marginal free; this source-anchored relaxation is what allows the penalty $\Phi$ to reshape the target endpoint proposal for subgroup consistency.
 Here, $\lambda$ is the weight that adjusts the penalty level, and $\Phi(X_0, X_1)$ scores how implausible a pairing $(X_0, X_1)$ is. The KL term keeps $P$ close to the independent bridge mixture, while the penalty term discourages endpoint pairs with high $\Phi$. As $\lambda$ increases, the solution increasingly favors path measures whose endpoint pairs have low penalty $\Phi$, i.e.\ biologically consistent pairings. We next show that this regularized objective admits an implementation as standard bridge matching under a Feynman-Kac reweighted reference bridge measure.

\subsection{Feynman-Kac Reweighting of the Reference Bridge Measure}
\label{subsec:fk_reweighting}

The source-anchored endpoint penalty in Eq.~\eqref{eq:endpoint-regularized-objective} can be absorbed into the reference path measure through a Feynman-Kac-type exponential reweighting. Classical Feynman-Kac theory describes how a reference stochastic process can be modified by multiplying its path measure by an exponential potential functional \cite{moral2004feynman}. Here the potential is an endpoint term $\Phi(X_0,X_1)$ encoding biological compatibility between a source and target sample, not a running cost accumulated along the path. This endpoint-potential reweighting changes the endpoint coupling used to construct the bridge mixture, while leaving the conditional Brownian bridge law between any fixed endpoint pair unchanged. 

Define the source-anchored Feynman-Kac reweighted reference
$R_\lambda$ via the Radon--Nikod\'{y}m derivative\cite{billingsley1995probability}:

\begin{equation}
\label{eq:tilted-reference}
\frac{dR_\lambda}{dR}(\omega)
= \frac{
  \exp\!\big(\!-\lambda\Phi(X_0(\omega), X_1(\omega))\big)
}{Z_\lambda(X_0(\omega))},
\end{equation}
where the source-conditional normalizing constant is
\begin{equation}
\label{eq:partition}
Z_\lambda(x_0) := \mathbb{E}_{X_1\sim\pi_1}\!\big[
  \exp\!\big(\!-\lambda\Phi(x_0, X_1)\big)
\big].
\end{equation}

The factor $e^{-\lambda\Phi}/Z_\lambda(x_0)$ reweights whole trajectories by endpoint compatibility, with $Z_\lambda(x_0)$ renormalizing for each $x_0$, so the source marginal stays $\pi_0$ and only the target endpoint is reshaped. Because $\Phi(X_0,X_1)$ depends on the endpoints alone rather than on a running path cost, the reweighting changes the coupling and leaves the path dynamics unchanged. We formalize its properties in the following propositions.

\medskip
\begin{proposition}[Endpoint penalty as reweighting]
\label{prop:kl-equivalence}
For any $P \ll R$ with $P_0=\pi_0$,
\begin{align}
\label{eq:kl-tilt-identity}
\mathrm{KL}(P \,\|\, R_\lambda)
&= \mathrm{KL}(P \,\|\, R)
  + \lambda\,\mathbb{E}_{P}[\Phi(X_0, X_1)] \notag\\
&\quad + \mathbb{E}_{X_0\sim\pi_0}[\log Z_\lambda(X_0)].
\end{align}
\textup{Consequently, minimizing the regularized objective in Eq.~\eqref{eq:endpoint-regularized-objective} is equivalent to minimizing $\mathrm{KL}(P \,\|\, R_\lambda)$.}
\end{proposition}

\begin{proof}
By the Radon--Nikod\'{y}m chain rule and Eq.~\eqref{eq:tilted-reference}, $\log\frac{dP}{dR_\lambda}=\log\frac{dP}{dR}+\lambda\Phi(X_0,X_1)+\log Z_\lambda(X_0)$. Taking $\mathbb{E}_P$ and using $P_0=\pi_0$ for the last term gives Eq.~\eqref{eq:kl-tilt-identity}; since $\mathbb{E}_{\pi_0}[\log Z_\lambda]$ is independent of the conditional proposal, minimizing $\mathrm{KL}(P\|R_\lambda)$ is equivalent to Eq.~\eqref{eq:endpoint-regularized-objective}.
\end{proof}

\medskip
\begin{proposition}[Reweighted endpoint proposal]
\label{prop:tilted-coupling}
\textup{The endpoint distribution under $R_\lambda$ has source marginal
$\pi_0$ and target conditional}
\begin{equation}
\label{eq:tilted-endpoint-coupling}
r^{\lambda}_{1|0}(x_1\mid x_0)
=
\frac{
e^{-\lambda\Phi(x_0, x_1)}\,\pi_1(x_1)
}{
Z_\lambda(x_0)
}.
\end{equation}
\end{proposition}

\begin{proof}
Integrating $dR_\lambda/dR$ from Eq.~\eqref{eq:tilted-reference} over a measurable endpoint set, with $(X_0,X_1)\sim\pi_0\otimes\pi_1$ under $R$ and $Z_\lambda$ depending only on $x_0$, yields Eq.~\eqref{eq:tilted-endpoint-coupling}.
\end{proof}

Proposition~\ref{prop:tilted-coupling} shows the reweighting reshapes only the \emph{conditional endpoint proposal}, favoring pairings with low penalty $\Phi$; for example, under a tau-positivity mismatch penalty it preferentially pairs same-status subjects across sites.

\medskip
\begin{proposition}[Bridge dynamics are unchanged]
\label{prop:bridge-unchanged}
\textup{For any endpoints} $(x_0, x_1)$,
\begin{equation}
\label{eq:conditional-bridge-unchanged}
R_\lambda(\cdot \mid X_0\!=\!x_0, X_1\!=\!x_1)
= \mathcal{W}^{x_0 \to x_1}(\cdot).
\end{equation}
\end{proposition}

\begin{proof}
Since $dR_\lambda/dR\propto e^{-\lambda\Phi(X_0,X_1)}$ depends only on the endpoints, it is constant given $(X_0,X_1)=(x_0,x_1)$ and cancels in the conditional law, so $R_\lambda(\cdot\mid x_0,x_1)=R(\cdot\mid x_0,x_1)=\mathcal{W}^{x_0\to x_1}(\cdot)$ by Eq.~\eqref{eq:indep-bridge-mixture}.
\end{proof}

Proposition~\ref{prop:bridge-unchanged} is what makes the method practical. The reweighting changes \emph{which} endpoint pairs are drawn but leaves the Brownian-bridge dynamics and hence the regression target $v_{\mathrm{BB}}$ unchanged.

\subsection{Feynman-Kac Reweighted DSBM: Objective, Algorithm, and Inference}
\label{subsec:fk_dsbm}

\begin{algorithm}[th!]
    \caption{Feynman-Kac Reweighted Schr\"odinger Bridge Matching}
    \label{alg:dsbm_harmonization}
    \SetAlgoLined
    \KwIn{Source and target datasets $\mathcal{D}_0,\mathcal{D}_1$;
      tau-positivity labels $g$; diffusion scale $\epsilon$;
      reweighting strength $\lambda$; pretrain steps $N_{\text{pre}}$;
      finetune steps $N_{\text{fine}}$; batch size $B$; EMA decay $\gamma$.}

    Initialize forward and backward drifts
    $v_\theta^{f}, v_\theta^{b}$;\quad
    $\theta_{\text{EMA}} \leftarrow \theta$\;

    \BlankLine
    \tcp{Stage 1: Pretrain}
    \For{$n = 1, \dots, N_{\text{pre}}$}{
        Draw source batch $X_0^{1:B}$ from $\mathcal{D}_0$\;
        For each $X_0^i$, sample $X_1^i$ from
        $q_\lambda(\cdot\mid X_0^i)$ in
        Eq.~\eqref{eq:conditional-tilted-proposal}\;
        $t \sim \mathrm{Unif}([0,1])^{\otimes B}$;\quad
        $Z \sim \mathcal{N}(0, I_d)^{\otimes B}$\;
        $X_t \leftarrow (1-t)X_0+tX_1+
        \sqrt{\epsilon t(1-t)}Z$\;
        Update $v_\theta^f$ and $v_\theta^b$ using the unweighted
        bridge-matching MSE to the analytic forward and backward
        Brownian-bridge drifts\;
        $\theta_{\text{EMA}} \leftarrow
          \gamma\,\theta_{\text{EMA}} + (1\!-\!\gamma)\,\theta$\;
    }

    \BlankLine
    \tcp{Stage 2: Finetune}
    \For{$n = 1, \dots, N_{\text{fine}}$}{
        Draw reweighted endpoint batch $(X_0^{1:B},X_1^{1:B})$ as above\;
        $\widehat{X}_1 \leftarrow$ forward SDE from $X_0$
          using $v^f_{\theta_{\text{EMA}}}$
          (Eq.~\eqref{eq:sde-inference})\;
        $\widehat{X}_0 \leftarrow$ backward SDE from $X_1$
          using $v^b_{\theta_{\text{EMA}}}$\;
        $t \sim \mathrm{Unif}([0,1])^{\otimes B}$;\quad
        $Z \sim \mathcal{N}(0, I_d)^{\otimes B}$\;
        $X_t^{(a)} \leftarrow
          \mathrm{Interp}_t(\widehat{X}_0, X_1, Z)$\;
        $X_t^{(b)} \leftarrow
          \mathrm{Interp}_t(X_0, \widehat{X}_1, Z)$\;
        Update $v_\theta^f$ on bridges
        $(\widehat{X}_0,X_1)$ and $v_\theta^b$ on bridges
        $(X_0,\widehat{X}_1)$ using unweighted MSE\;
        $\theta_{\text{EMA}} \leftarrow
          \gamma\,\theta_{\text{EMA}} + (1\!-\!\gamma)\,\theta$\;
    }

    \BlankLine
    \tcp{Inference}
    Given $x_0 \sim \pi_0$, simulate SDE
    (Eq.~\eqref{eq:sde-inference}) with $v_{\theta_{\text{EMA}}}$
    from $t\!=\!0$ to $t\!=\!1$ to obtain harmonized
    $\hat{x}_1$\;

    \KwOut{Parameters $(\theta, \theta_{\text{EMA}})$}
\end{algorithm}

Section~\ref{subsec:fk_reweighting} establishes the central result of our method, that the endpoint-regularized bridge is equivalent to an unregularized Schr\"odinger bridge against the reweighted reference $R_\lambda$. Proposition~\ref{prop:kl-equivalence} proves that minimizing the regularized objective of Eq.~\eqref{eq:endpoint-regularized-objective} is identical to solving this reweighted bridge, so the biological constraint is enforced through the reference measure rather than through the loss or the network. Proposition~\ref{prop:tilted-coupling} redirects the whole endpoint coupling toward biologically compatible pairs in Eq.~\eqref{eq:tilted-endpoint-coupling}, and Proposition~\ref{prop:bridge-unchanged} guarantees that the Brownian-bridge law of every pair, and hence the regression target $v_{\mathrm{BB}}$, is preserved exactly. The regularized transport is therefore solved by the same bridge-matching iteration acting on the reweighted coupling, and the learned drift realizes the regularized bridge.


It remains to draw endpoint pairs from $R_\lambda$. One realization keeps independent product draws $(X_0,X_1)\sim\pi_0\otimes\pi_1$ and corrects each pair by the importance weight $\tilde{w}=\exp(-\lambda\Phi(X_0,X_1))$, giving
the self-normalized importance-sampling objective $\mathcal{L}_\lambda(\theta)$ = 
\begin{equation}
\label{eq:tilted-dsbm-weighted}
\mathbb{E}_{X_0\sim\pi_0}\!\left[
\frac{
  \mathbb{E}_{X_1\sim\pi_1}\!\Big[
    \tilde{w}(X_0,X_1)\,
    \mathbb{E}_{t, X_t | X_0, X_1}\!\big[\|v_\theta - v_{\mathrm{BB}}\|_2^2\big]
  \Big]
}{
  \mathbb{E}_{X_1\sim\pi_1}[\tilde{w}(X_0,X_1)]
}
\right].
\end{equation}
This estimator is correct but wasteful when $\lambda$ is large since mismatched pairs receive vanishing weight yet still incur a full forward and backward pass. Our implementation therefore absorbs the weights into the sampler, drawing each target endpoint from the empirical conditional proposal
\begin{equation}
\label{eq:conditional-tilted-proposal}
q_\lambda(x_1 \mid x_0)
=
\frac{
  \exp[-\lambda\Phi(x_0,x_1)]\,\hat{\pi}_1(x_1)
}{
  \sum_{x'_1 \in \mathcal{D}_1}
  \exp[-\lambda\Phi(x_0,x'_1)]\,\hat{\pi}_1(x'_1)
},
\end{equation}
where $\mathcal{D}_1$ is the target training set and $\hat{\pi}_1$ its
empirical measure; $q_\lambda$ is exactly the empirical counterpart of the
coupling in Proposition~\ref{prop:tilted-coupling}. The bridge-matching loss
is then evaluated without explicit per-sample weights, since the pairs are
already drawn from the reweighted proposal, concentrating computation on pairs
with non-negligible weight rather than on pairs whose gradients are nearly
zero. The realized objective is
\begin{equation}
\label{eq:proposal-loss}
\mathcal{L}_\lambda(\theta)
=
\mathbb{E}_{X_0\sim\hat{\pi}_0}
\mathbb{E}_{X_1\sim q_\lambda(\cdot\mid X_0)}
\!\Big[
\mathbb{E}_{t, X_t | X_0, X_1}\!\big[\|v_\theta - v_{\mathrm{BB}}\|_2^2\big]
\Big],
\end{equation}
which is the configuration used for the reported FKRSBM experiments. Eqs.~\eqref{eq:tilted-dsbm-weighted} and~\eqref{eq:proposal-loss} are two estimators of the same reweighted expectation, so sampling from $q_\lambda$ optimizes the same objective as importance weighting, only more efficiently.

The penalty $\Phi(x_0,x_1)$ encodes which cross-distribution pairings are
meaningful. Because it enters only through $\tilde{w}=\exp(-\lambda\Phi)$ and
receives no gradients, its design is flexible. For the tau PET task, where the
cohorts differ in tau-positivity prevalence, we use a mismatch penalty
\begin{equation}
\label{eq:phi-design}
\Phi(x_0, x_1) = \mathbf{1}[g(x_0) \neq g(x_1)],
\end{equation}
where $g(\cdot)\in\{\tau^{+},\tau^{-}\}$ is the tau-positivity label obtained by thresholding mean cortical SUVR with the standard criterion of Section~\ref{subsec:dataset_preprocess}. Same-group pairs then receive weight $\tilde{w}=1$ and cross-group pairs $\tilde{w}=e^{-\lambda}$, steering the transport to align each tau-positivity subgroup with its target counterpart rather than with the target marginal as a whole. More generally, $\Phi$ may combine multiple covariates as a weighted sum of mismatch terms or encode feature-space distances, without modifying the underlying transport algorithm.

With the sampler fixed, training follows the two-stage DSBM recipe~\cite{alpha_dsbm}. DSBM maintains a forward drift $v_\theta^{f}$ ($\pi_0\!\to\!\pi_1$) and a backward drift $v_\theta^{b}$ ($\pi_1\!\to\!\pi_0$), trained in alternation; $v_\theta$ denotes the forward drift used at inference. In the \emph{pretraining} stage, endpoint pairs are drawn from Eq.~\eqref{eq:conditional-tilted-proposal} and the drift is regressed onto Brownian-bridge velocities. In the \emph{finetuning} stage, the model generates self-consistent endpoints by integrating the learned SDE forward and backward from samples drawn by the same proposal, and refines the drift on the resulting bridge-matching targets; an exponential moving average
(EMA) of the parameters is maintained for stable trajectory generation. Algorithm~\ref{alg:dsbm_harmonization} summarizes the full procedure.

At inference, the trained model harmonizes a source sample $x_0\sim\pi_0$ by simulating the learned dynamics from $t=0$ to $t=1$. Because the Schr\"odinger bridge solution is itself a diffusion (Eq.~\eqref{eq:ito-sde}),
we retain the stochastic dynamics and solve
\begin{equation}
\label{eq:sde-inference}
dx_t = v_\theta(t, x_t)\,dt + \sqrt{\epsilon}\,dW_t,
\quad x_0 \sim \pi_0,
\quad t \in [0, 1],
\end{equation}
taking the harmonized output as $\hat{x}_1=x_{t=1}$. We integrate with the Euler--Maruyama scheme~\cite{kloeden1992stochastic} using $N$ uniformly spaced steps of size $h=1/N$:
\begin{equation}
\label{eq:euler-maruyama-step}
x_{t_{k+1}} = x_{t_k} + h \cdot v_\theta(t_k, x_{t_k})
+ \sqrt{\epsilon h}\;\xi_k,
\quad \xi_k \sim \mathcal{N}(0, I_d),
\end{equation}
with $t_k=k/N$. In all experiments we use $N=100$ steps, a stable tradeoff between integration accuracy and computational cost.

\subsection{Implementation}
The backbone network is a conditional surface U-Net (CUNet) composed of spherical convolution layers~\cite{Spherical_unet} arranged in an encoder--decoder architecture on the sixth-order icosahedral mesh. We use channel widths $[32,64,128,256]$, group norm, ReLU in convolutional blocks, and sinusoidal time embeddings injected into residual blocks. No auxiliary clinical condition is provided to the CUNet; the endpoint reweighting is implemented entirely through the tau-positivity-aware sampler. Before input to the network, a log transform is applied to the data to suppress abnormally high SUVR values. We set diffusion scale $\epsilon=0.01$ and reweighting strength $\lambda=4$ by grid search.  The model is trained at batch size 16 in pretraining for 20,000 iterations and batch size 4 in finetuning for 10,000 iterations. Both inference and finetune use 100 sampling steps, SDE inference, EMA decay 0.999, and surface smoothing with $\sigma=0.05$. The learning rates used in pretraining and finetuning are $10^{-4}$ and $10^{-5}$, respectively. The best model is selected based on validation-set performance metrics as in Table \ref{tab:tau_harmonization}. All experiments were conducted on an NVIDIA RTX A6000 GPU with 48 GB of memory.

\section{Results}
\subsection{Dataset and Preprocessing}
\label{subsec:dataset_preprocess}

\begin{figure}[t]
    \centering
    \includegraphics[width=\linewidth]{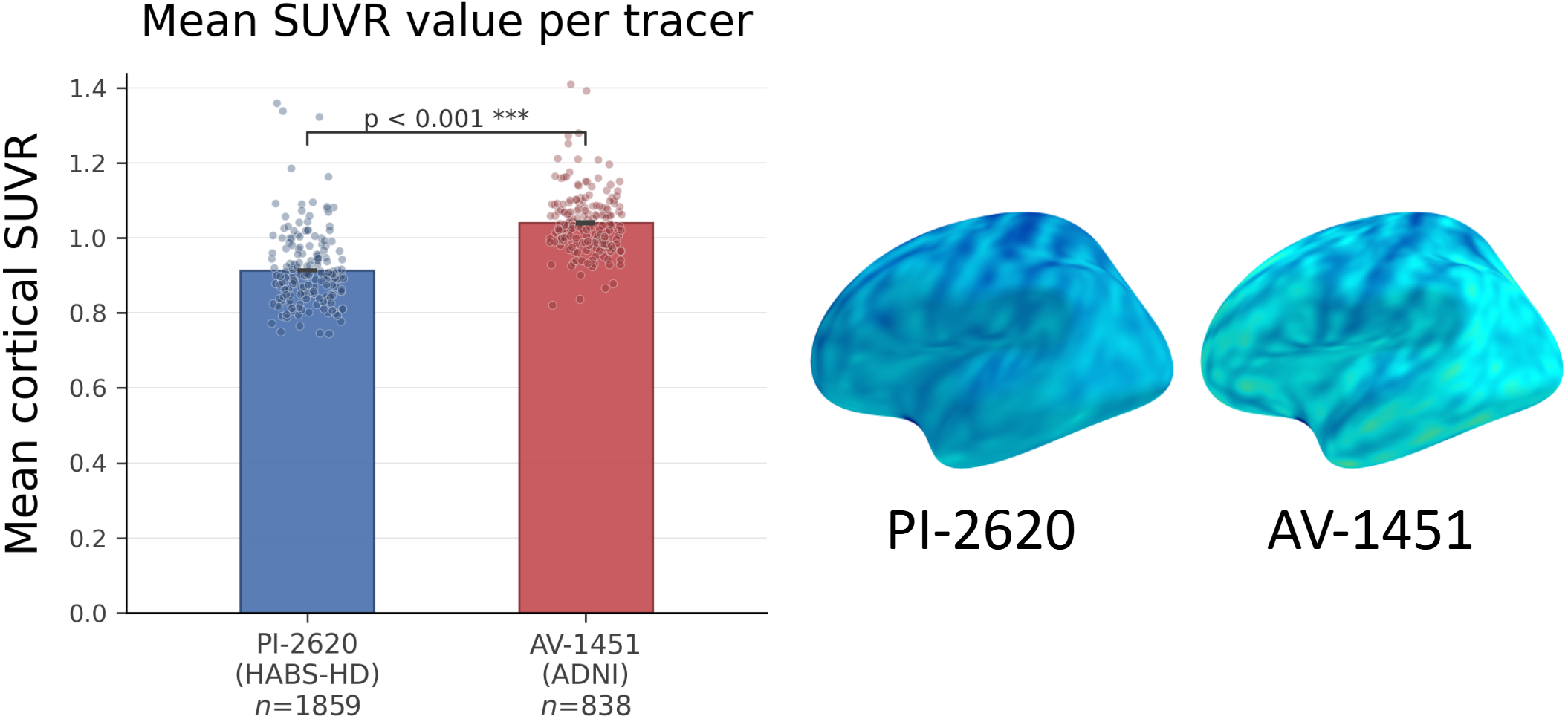}
    \caption{\textbf{Distributional and spatial differences between tau PET tracers.}
    \textit{Left:} Bar plots of scan-level mean left hemisphere cortical SUVR for PI-2620 and AV-1451 in cognitively normal (CN) subjects show a significant shift in central tendency (Mann--Whitney $U$, $p < 0.001$), with AV-1451 yielding systematically higher SUVR values. \textit{Right:} Cohort-averaged SUVR projected onto the lateral cortical surface reveals tracer-specific spatial signal patterns, motivating the need for harmonization across tracers.}
    \label{fig:dataset}
\end{figure}

For the tau PET harmonization experiments, we utilized two public datasets: the Health and Aging Brain Study-Health Disparities (HABS-HD)~\cite{petersen2025health} and the Alzheimer's Disease Neuroimaging Initiative (ADNI)~\cite{Mueller05}. Each acquisition comprised six consecutive 5-minute frames, coregistered to the first frame and averaged. The averaged PET and the temporally closest T1w MRI (processed by FreeSurfer~\cite{fischl2012freesurfer}) are analyzed with PETSurfer~\cite{greve2014cortical}, and SUVR is computed relative to the mean inferior cerebellum uptake and projected onto the subject-specific cortical surface via \textit{mri\_vol2surf}. Because off-target binding from extracerebral tissue (e.g., meninges) contaminates cortical signal, we apply artifact removal~\cite{yue2025tau}. The cortical SUVR map is then resampled to a sixth-order icosahedral mesh through the \textit{mri\_surf2surf} command. Tau-positivity status is assigned by thresholding the mean cortical SUVR using the method in ~\cite{JACK2018535}. Distributional differences are shown in Figure~\ref{fig:dataset} where AV1451 tracer generally yields higher SUVR than PI2620. The statistics of datasets are shown in Table \ref{tab:tau_demographic}.

\begin{table}[h!]
\centering
\caption{Demographics, scan counts, and data splits for the HABS-HD and ADNI tau PET datasets used in the SUVR harmonization experiment. CN: Cognitively Normal; MCI: Mild Cognitive Impairment; AD: Alzheimer’s Disease. Train/validation/test splits were defined at the subject level (6:2:2); scan counts are reported below.} 
\begin{tabular}{lcc}
\hline
\textbf{Attribute} & \textbf{HABS-HD} & \textbf{ADNI} \\ \hline
N (Subjects) 
& 1797 
& 826 \\

Number of Scans
& 2458
& 1480 \\

Tracer Types
& PI-2620 
& AV-1451 \\

Age (Mean $\pm$ SD, years) 
& $66.2 \pm 8.7$ 
& $74.3 \pm 8.1$ \\


Diagnosis (CN / MCI / AD scans) 
& 1859 / 486 / 113 
& 838 / 482 / 160 \\

Train / Val / Test scans
& 1458 / 497 / 503 
& 896 / 289 / 295 \\ \hline
\end{tabular}

\label{tab:tau_demographic}
\end{table}

\begin{table*}[!th]
\centering
\caption{Quantitative comparison of harmonization methods. Metrics include Wasserstein Distance (↓), PCC (↑), absolute Somers' D of cohort classification (↓), and Wasserstein Distance within tau-positive and tau-negative subgroups (↓). Somers' D = 2$\cdot$AUC$-$1; values near 0 indicate site-indistinguishable harmonization. Bold indicates best.}
\resizebox{\textwidth}{!}{%
\begin{tabular}{lccccc|ccccc}
\hline
& \multicolumn{5}{c}{Left Hemisphere} & \multicolumn{5}{c}{Right Hemisphere} \\ \hline
\textbf{Method} 
& WD$\downarrow$ & PCC$\uparrow$ & $|D|\downarrow$ & WD $\tau_{pos}\downarrow$ & WD $\tau_{neg}\downarrow$
& WD$\downarrow$ & PCC$\uparrow$ & $|D|\downarrow$ & WD $\tau_{pos}\downarrow$ & WD $\tau_{neg}\downarrow$ \\ \hline
ComBat & 0.0619 & 0.9517 & 0.3314 & 0.2296 & 0.0493 & 0.0357 & \textbf{0.9605} & 0.1824 & 0.1076 & 0.0327 \\
CycleGAN & 0.0601 & 0.9468 & 0.0700 & 0.3628 & 0.1578 & 0.0631 & 0.9547 & 0.6140 & 0.1430 & 0.0942 \\
DF & 0.1254 & 0.9244 & 0.7600 & 0.4279 & 0.2069 & 0.1946 & 0.9414 & 0.3960 & 0.3705 & 0.3685 \\ 
DSBM & 0.0496 & 0.9359 & 0.1694 & 0.1187 & 0.0627 & 0.0285 & 0.9484 & 0.1520 & 0.0935 & 0.0449 \\
FKRSBM & \textbf{0.0079} & \textbf{0.9519} & \textbf{0.0696} & \textbf{0.0486} & \textbf{0.0043} & \textbf{0.0100} & 0.9503 & \textbf{0.0504} & \textbf{0.0376} & \textbf{0.0072} \\ \hline
\end{tabular}
}

\label{tab:tau_harmonization}
\end{table*}

\begin{figure*}[!t]
  \centering
  \includegraphics[width=1.6\columnwidth]{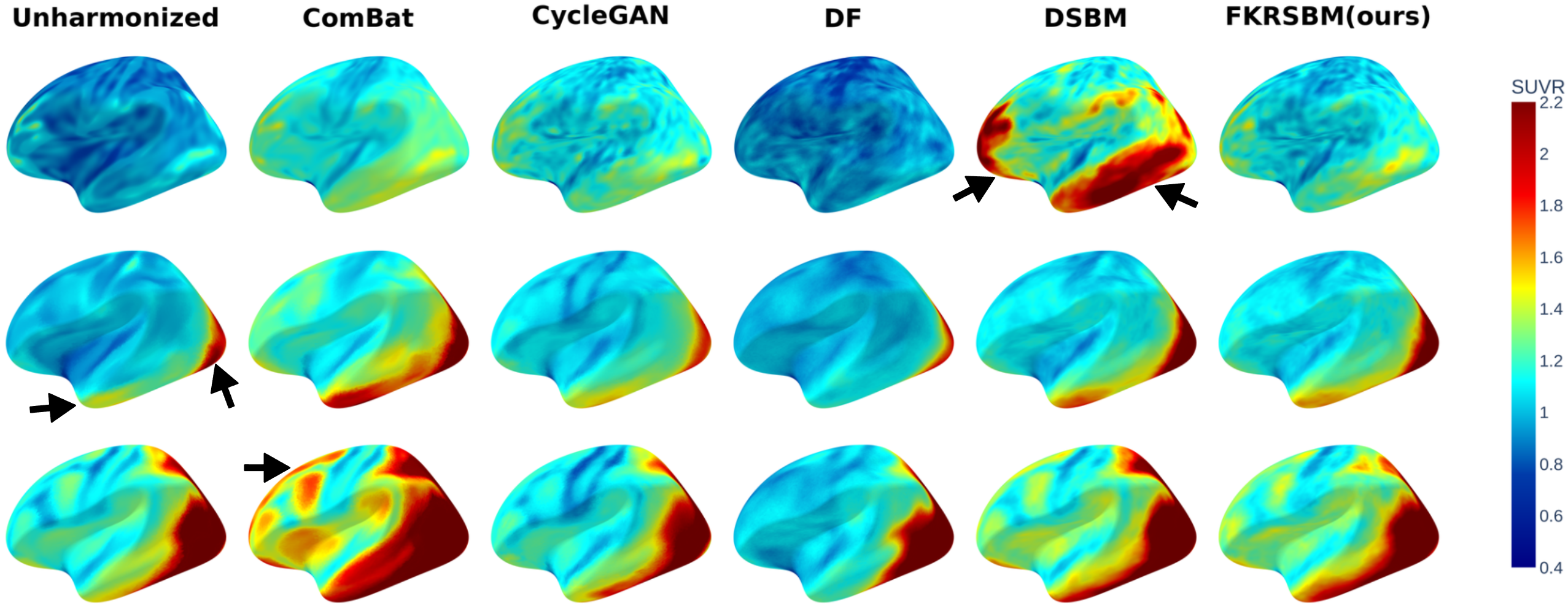}
  \caption{Harmonization visualization for three representative HABS-HD subjects and harmonized outputs produced by ComBat, CycleGAN, DF, DSBM and FKRSBM. The rows illustrate examples with different tau pathology burdens including clear temporal tau pathology and lower-SUVR cases with subtler temporal elevation.}
  \label{fig:tau_harmonization_example}
\end{figure*}

\subsection{Tau PET SUVR Harmonization Across Tracers}

We designate HABS-HD (PI-2620) as the source domain and ADNI (AV-1451) as the target domain, based on the standardized diagnostic protocols of ADNI and the widespread use of the AV-1451 tracer. The objective is to harmonize HABS-HD PI-2620 SUVR maps into the ADNI AV-1451 target domain. Both datasets are split into train/validation/test sets separately on subject level, as shown in Table \ref{tab:tau_demographic}. The model is trained on the training set and validated on the validation set for model selection. All reported metrics are computed on the held-out test set.

Quantitative results are shown in Table~\ref{tab:tau_harmonization}. We assess pattern preservation using the Pearson correlation (PCC) between each source map and its harmonized output, and target-domain alignment using Wasserstein distance (WD). Left and right hemispheres are evaluated separately. Residual domain separability is measured by the absolute Somers' D of an SVM cohort classifier on ROI mean features ($|D|=|2\cdot\text{AUC}-1|$, AUC=Area Under Curve), with lower values indicating stronger site-effect removal. We also compute WD within tau-positive and tau-negative groups to assess subgroup-level alignment. Overall, FKRSBM provides the most consistent harmonization across hemispheres, achieving stronger target-domain alignment and lower residual cohort separability while preserving subject-level spatial correspondence. CycleGAN shows unstable residual separability across hemispheres, whereas DF performs poorly in this setting, likely reflecting the tradeoff between noise-based cohort information removal and spatial pattern preservation. FKRSBM further achieves the lowest WD within both the tau-positive and tau-negative groups, showing the gain holds within tau subgroups and not only at the marginal.

Figure~\ref{fig:tau_harmonization_example} presents harmonized SUVR maps for three representative subjects, with black arrows indicating the regions referenced below. Consistent with the tracer difference illustrated in Figure~\ref{fig:dataset}, harmonization from PI-2620 to AV-1451 is expected to increase cortical SUVR, and we therefore assess each method by whether this adjustment is applied proportionately while preserving the relative contrast of different cortical regions  of the source. For the low-burden subject (top row), DSBM amplifies the adjustment disproportionately, introducing pronounced frontal and temporal elevations that are inconsistent with the source signal. In the middle row, the arrows indicate two  regions of differing tau intensity, and FKRSBM maintains their relative contrast, whereas ComBat, CycleGAN, and DF attenuate or distort this relationship. The bottom row shows a subject with pronounced pathology, where ComBat similarly overestimates the adjustment and produces a diffuse, spatially broad elevation, whereas FKRSBM scales the signal into the target range while preserving the temporal-lateral distribution of the source.

\begin{figure*}[!t]
  \centering
  \includegraphics[width=2\columnwidth]{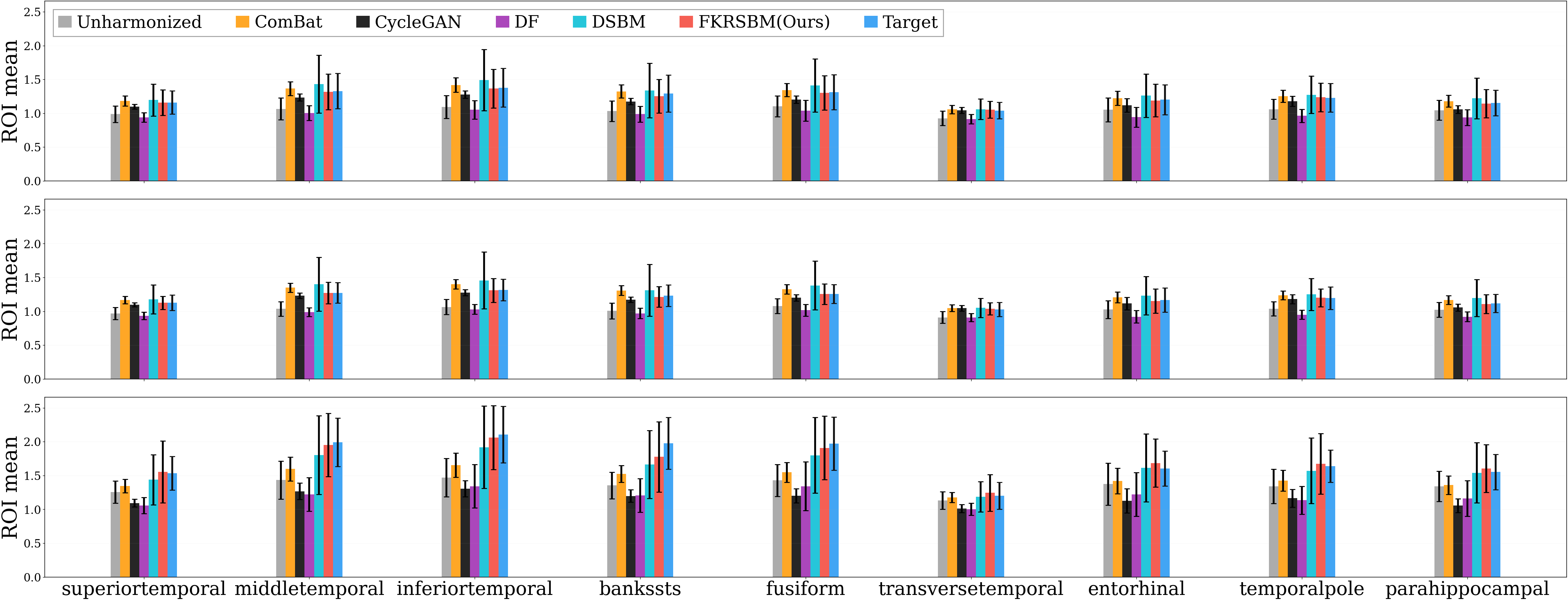}
  \caption{Region-level harmonization comparison in the temporal lobe across Desikan--Killiany ROIs. Each group of bars shows the mean SUVR for a single ROI in the temporal lobe, with error bars denoting standard deviation. The three rows correspond to the overall test set (top), the tau-negative subgroup (middle), and the tau-positive subgroup (bottom). }
  \label{fig:roi_bar}
\end{figure*}

Figure~\ref{fig:roi_bar} compares temporal-lobe ROI mean SUVR across the overall test set, tau-negative subjects, and tau-positive subjects. In the overall and tau-negative groups, most methods except DF partially reduce the source-target mean shift. However, ComBat and CycleGAN mainly align the mean while failing to reproduce the target variance, with CycleGAN showing clear variance compression across many ROIs. DF shows poorer alignment, with shifted means and reduced regional fidelity. The limitation of the baselines is more pronounced in the tau-positive subgroup, where most methods fail to match both the elevated target means and the larger variance observed in ADNI. DSBM improves over these baselines in some ROIs but remains inconsistent across subgroups, reflecting the product-coupling failure mode of Sec.~\ref{subsec:endpoint_regularized}. FKRSBM tracks both the target mean and standard deviation across Temporal ROIs, preserving subgroup-specific tau variability.

\subsection{Tau Positivity Preservation}
\begin{table*}[!t]
\centering
\caption{Tau positivity sign mismatch analysis after harmonization. pos$\rightarrow$neg and neg$\rightarrow$pos indicate the direction of status flips.}
\resizebox{\textwidth}{!}{%
\begin{tabular}{lccccc|ccccc}
\hline
& \multicolumn{5}{c}{Left Hemisphere} & \multicolumn{5}{c}{Right Hemisphere} \\ \hline
\textbf{Method} 
& n\_total & n\_mismatch & mismatch\_pct & pos$\rightarrow$neg & neg$\rightarrow$pos
& n\_total & n\_mismatch & mismatch\_pct & pos$\rightarrow$neg & neg$\rightarrow$pos \\ \hline
ComBat & 503 & 33 & 6.6\% & 0 & 33 & 503 & 26 & 5.2\% & 0 & 26 \\
CycleGAN & 503 & 36 & 7.2\% & 36 & 0 & 503 & 34 & 6.8\% & 34 & 0 \\
DF & 503 & 35 & 7.0\% & 35 & 0 & 503 & 38 & 7.6\% & 38 & 0 \\
DSBM & 503 & 88 & 17.5\% & 18 & 70 & 503 & 67 & 13.3\% & 16 & 51 \\

FKRSBM & 503 & \textbf{13} & \textbf{2.6\%} & 9 & 4 & 503 & \textbf{15} & \textbf{3.0\%} & 4 & 11 \\ \hline
\end{tabular}
}

\label{tab:sign_mismatch}
\end{table*}

A clinically meaningful harmonization should align the tau-positivity status of individual subjects. After harmonizing the HABS-HD SUVR maps into the ADNI domain, we recompute tau positivity using the ADNI-derived cutoff and compare the post-harmonization label with the original pre-harmonization label. A status flip (tau-positive becoming negative, or vice versa) indicates that the harmonization has altered the biological signal sufficiently to change a subject's clinical categorization. Table~\ref{tab:sign_mismatch} reports the results for 503 test scans in each hemisphere. FKRSBM has the lowest mismatch rate in this comparison: 2.6\% (13/503) in the left hemisphere and 3.0\% (15/503) in the right hemisphere. The direction of flips is also balanced, with roughly comparable numbers of pos$\rightarrow$neg and neg$\rightarrow$pos errors, suggesting no systematic bias in either direction. The DSBM produces the highest mismatch rate (17.5\% left, 13.3\% right), and its flips are asymmetric, mostly neg$\rightarrow$pos, consistent with the product-coupling failure mode of Sec.~\ref{subsec:endpoint_regularized} that Feynman-Kac reweighting is designed to reduce. ComBat shows moderate mismatch (6.6\%/5.2\%), all neg$\rightarrow$pos, reflecting its upward shift; CycleGAN is comparable (7.2\%/6.8\%) but exclusively pos$\rightarrow$neg. FKRSBM's low and balanced mismatch rate shows that Feynman-Kac reweighted transport better preserves tau-positivity subgroup structure during harmonization.

\subsection{Downstream Clinical Assessment}

\begin{figure*}[!t]
  \centering
  \includegraphics[width=2\columnwidth]{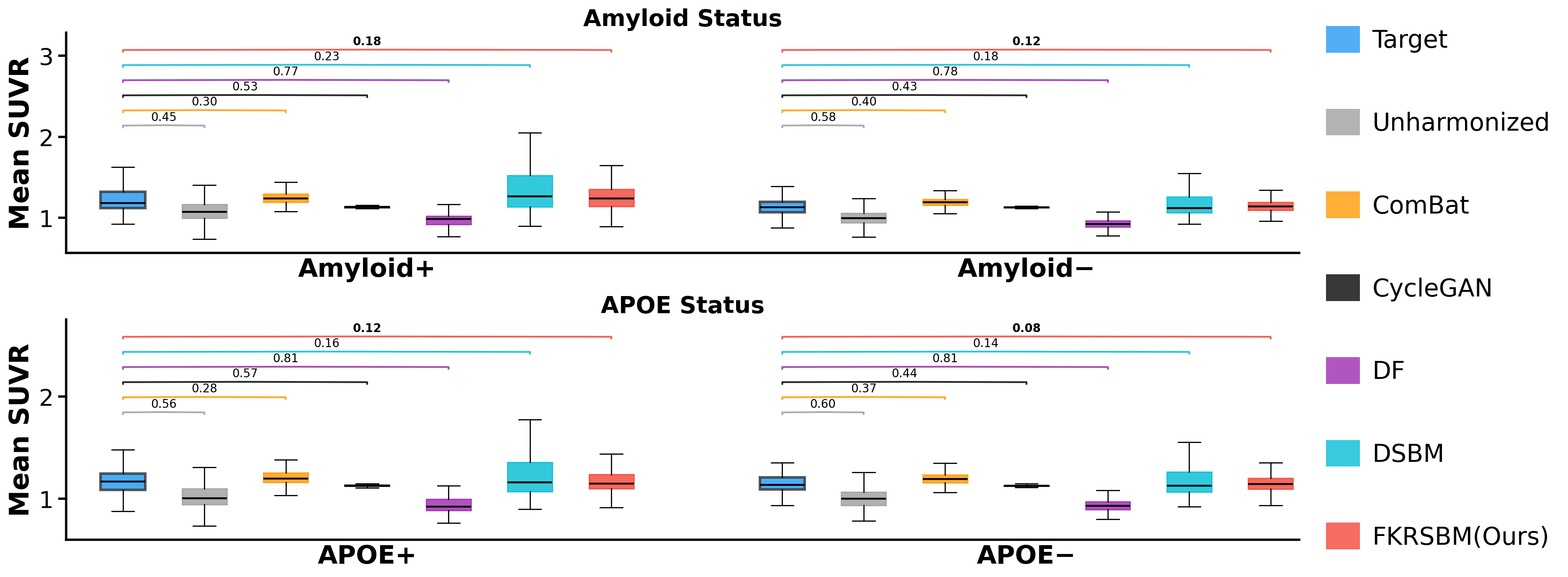}
  \caption{Subgroup alignment of left-hemisphere mean tau SUVR by amyloid and APOE $\varepsilon$4 status. The top row shows amyloid-positive and amyloid-negative subgroups, and the bottom row shows APOE $\varepsilon$4 carriers and non-carriers. Within each subgroup, boxplots compare the target ADNI distribution with unharmonized HABS-HD data and the outputs from each harmonization method. The numbers above the colored brackets report the Kolmogorov--Smirnov (KS) statistic between each method and the target distribution; lower values indicate better distributional alignment.}
  \label{fig:apoe_suvr}
\end{figure*}

\begin{table*}[!th]
\caption{Downstream diagnostic separability on the combined test set. Best result per column in \textbf{bold}.}
\label{tab:classification_combined}
\centering
\scriptsize
\resizebox{\textwidth}{!}{
\begin{tabular}{lcccccccc}
\hline
\multirow{2}{*}{Method} 
& \multicolumn{4}{c}{CN vs AD} 
& \multicolumn{4}{c}{CN vs MCI vs AD} \\
\cline{2-9}
& F1 & Acc. & Spec. & Sens.
& F1 & Acc. & Spec. & Sens. \\
\hline

Unharmonized 
& $0.7404{\pm}0.0275$ & $0.7657{\pm}0.0186$ & $0.8625{\pm}0.0123$ & $0.6871{\pm}0.0398$
& $0.5314{\pm}0.0301$ & $0.5410{\pm}0.0276$ & $0.7739{\pm}0.0129$ & $0.5474{\pm}0.0267$ \\

ComBat 
& $0.7426{\pm}0.0115$ & $0.7690{\pm}0.0121$ & $0.8649{\pm}0.0333$ & $0.6846{\pm}0.0174$
& $0.5122{\pm}0.0333$ & $0.5220{\pm}0.0325$ & $0.7629{\pm}0.0156$ & $0.5258{\pm}0.0315$ \\

CycleGAN 
& $0.7373{\pm}0.0332$ & $0.7659{\pm}0.0283$ & $0.8743{\pm}0.0242$ & $0.6793{\pm}0.0367$
& $0.5175{\pm}0.0382$ & $0.5278{\pm}0.0383$ & $0.7675{\pm}0.0167$ & $0.5366{\pm}0.0360$ \\

DF
& $0.7360{\pm}0.0439$ & $0.7698{\pm}0.0302$ & $\mathbf{0.8877{\pm}0.0254}$ & $0.6692{\pm}0.0632$
& $0.4914{\pm}0.0422$ & $0.4953{\pm}0.0456$ & $0.7497{\pm}0.0225$ & $0.5009{\pm}0.0439$ \\

DSBM
& $0.7286{\pm}0.0456$ & $0.7541{\pm}0.0370$ & $0.8400{\pm}0.0349$ & $0.6839{\pm}0.0524$
& $0.5095{\pm}0.0427$ & $0.5212{\pm}0.0468$ & $0.7628{\pm}0.0213$ & $0.5244{\pm}0.0436$ \\

\textbf{FKRSBM} 
& $\mathbf{0.8019{\pm}0.0373}$ & $\mathbf{0.8169{\pm}0.0304}$ & $0.8832{\pm}0.0231$ & $\mathbf{0.7621{\pm}0.0477}$
& $\mathbf{0.5794{\pm}0.0639}$ & $\mathbf{0.5862{\pm}0.0602}$ & $\mathbf{0.7952{\pm}0.0289}$ & $\mathbf{0.5888{\pm}0.0596}$ \\
\hline
\end{tabular}
}
\end{table*}

Because tau-positivity is the covariate used in training, the subgroup-WD and sign-consistency results above could in principle reflect the training signal itself. To rule this out, we test generalization on covariates that are never used during training: Amyloid PET status and APOE $\varepsilon$4 carrier status, a major genetic risk factor for AD. We report the Kolmogorov–Smirnov (KS) statistic between scan-level mean cortical SUVR distributions. KS is nonparametric, making no Gaussian assumption on the data, and is sensitive to differences in both the mean and the variance, so it reflects whether a method matches the full target distribution. As shown in Figure~\ref{fig:apoe_suvr}, FKRSBM most consistently matches the target ADNI distributions across both amyloid and APOE $\varepsilon$4 subgroups, achieving the lowest KS statistic in every panel, with the unregularized DSBM as the consistent runner-up. CycleGAN and DF show larger subgroup discrepancies, with DF in particular shifting SUVR distributions away from the target. FKRSBM also preserves the expected biological ordering, with amyloid-positive and APOE $\varepsilon$4-carrier subgroups showing higher mean SUVR than their respective negative subgroups after harmonization, consistent with the ADNI target pattern. Thus Feynman-Kac reweighting reduces tracer-related cohort differences without collapsing biological variation that was not used in training.

We further evaluate whether the harmonized data improves pooled-domain diagnostic separability. For each hemisphere, the surface SUVR map is reduced to a feature vector of Desikan--Killiany ROI means. A support vector machine (SVM) classifier is trained on the combined (harmonized HABS-HD $\cup$ ADNI) test set for two tasks:  CN vs.\ AD and the three-class CN vs.\ MCI vs.\ AD problem. This probes downstream statistical power, since harmonization that removes tracer variability while retaining disease signal should improve pooled-data separation. Because the diagnostic groups are imbalanced, we subsample each class to match the size of the smallest class. We repeat over 5 seeds and report mean$\pm$SD of F1, accuracy, specificity, and sensitivity from subject-level 3-fold cross-validation (preventing leakage across a subject's scans). FKRSBM yields the highest F1 on both tasks (Table~\ref{tab:classification_combined}), indicating better support for downstream clinical tasks.

\section{Conclusion}
We presented Feynman-Kac Reweighted Schr\"odinger Bridge Matching (FKRSBM) for unpaired cross-tracer tau PET harmonization on cortical surfaces. By combining direct Schr\"odinger Bridge transport with subgroup-aware endpoint reweighting, FKRSBM reduces tracer-related shifts while better preserving tau-pathology strata. Its surface-based design enables vertex-level harmonization that respects cortical geometry. In harmonizing HABS-HD PI-2620 data to the ADNI AV-1451 domain, FKRSBM improved subgroup alignment, tau-positivity consistency, and downstream diagnostic separability compared with baseline methods. These results demonstrate that endpoint-aware surface SB matching can improve biologically faithful cross-tracer tau PET harmonization.

\section{Acknowledgments}
\label{sec:acknowledgments}
The HABS-HD team include the MPIs: Sid E O’Bryant, Kristine Yaffe, Arthur Toga, Robert Rissman, and Leigh Johnson; and the HABS-HD Investigators: Meredith Braskie, Kevin King, James R Hall, Melissa Petersen, Raymond Palmer, Robert Barber, Yonggang Shi, Fan Zhang, Rajesh Nandy, Roderick McColl, David Mason, Bradley Christian, Nicole Philips and Stephanie Large. The HABS-HD project was supported by the National Institute on Aging of the National Institutes of Health under Award Numbers R01AG054073 and R01AG058533. The content is solely the responsibility of the authors and does not necessarily represent the official views of the National Institutes of Health. Data used in preparing this article were obtained from the ADNI database (\url{adni.loni.usc.edu}). As such, many investigators within the ADNI contributed to the design and implementation of ADNI and/or provided data but did not participate in analysis or writing of this report. A complete list of ADNI investigators: \url{http://adni.loni.usc.edu/wp-content/uploads/how_to_apply/ADNI_Acknowledgement_List.pdf.} 

\section*{REFERENCES}
\bibliographystyle{IEEEtran}
\bibliography{ref}

\appendices

\end{document}